\documentclass{article}

\usepackage{amsmath}
\usepackage{amsfonts}
\usepackage{amsthm}
\usepackage[letterpaper]{geometry}
\usepackage[utf8]{inputenc}
\usepackage{xkeyval}
\usepackage[ruled, vlined]{algorithm2e}

\newtheorem{lemma}{Lemma}
\newtheorem{theorem}{Theorem}

\usepackage{xspace}

\newcommand{\seth}{SETH\xspace}

\newcommand{\kcnf}[1][k]{$#1$-CNF\xspace}
\newcommand{\ksat}[1][k]{$#1$-SAT\xspace}
\newcommand{\maxksat}[1][k]{MAX-$#1$-SAT\xspace}
\newcommand{\cnfsat}{CNF-SAT\xspace}
\newcommand{\cnf}{CNF\xspace}
\newcommand{\dnf}{DNF\xspace}
\newcommand{\aczero}{\mathbf{AC^0}\xspace}
\newcommand{\acczero}{\mathbf{ACC^0}\xspace}
\newcommand{\tczero}{\mathbf{TC^0}\xspace}

\newcommand{\nats}{\mathbb{N}}

\newcommand{\ppoly}{\mathbf{P/Poly}\xspace}
\newcommand{\nexp}{\mathbf{NEXP}\xspace}
\newcommand{\ntime}{\mathbf{NTIME}\xspace}
\newcommand{\np}{\mathbf{NP}\xspace}
\newcommand{\p}{\mathbf{P}\xspace}

\title{Satisfiability Algorithms for Restricted Circuit Classes}
\author{Stefan Schneider}
\begin{document}
\date{May 10, 2013}

\maketitle
\begin{abstract}
  In recent years, finding new satisfiability algorithms for various
  circuit classes has been a very active line of research. Despite
  considerable progress, we are still far away from a definite answer
  on which circuit classes allow fast satisfiability algorithms. This
  survey takes a (far from exhaustive) look at some recent
  satisfiability algorithms for a range of circuit classes and
  highlights common themes. A special focus is given to connections
  between satisfiability algorithms and circuit lower bounds. A second
  focus is on reductions from satisfiability algorithms to a range of
  polynomial time problems, such as matrix multiplication and the
  Vector Domination Problem.
\end{abstract}

\section{Introduction}
Ever since Cook and Levin \cite{Cook1971, Levin1973} proved the
$\np$-completeness of \ksat[3], satisfiability problems played a
central role within complexity theory. Their result puts
satisfiability algorithms in the center of the arguably most important
question in computer science, $\p$ vs. $\np$.

To show $\p = \np$ with a satisfiability algorithm, the algorithm has
to run in polynomial time, which is a long way from the trivial
$\tilde{O}\left(2^n\right)$ exhaustive search algorithm. This raises
the question if we can improve over exhaustive search at all. In
particular, we want to know if we can find satisfiability algorithms
that improve over exhaustive search by an exponential factor,
i.e. run in time $\tilde{O}\left(2^{(1-\mu) n}\right)$ for some
constant $\mu > 0$. For such a runtime we refer to $\mu$ as the
\emph{savings} of the algorithm.

For \ksat the answer is yes due to Monien and Speckenmeyer
\cite{MonienSpeckenmeyer1985}. Since then, several faster algorithms
for \ksat have been found, e.g. \cite{PaturiPudlakZane1997,
  Schoning2002, PaturiPudlakSaksZane1998, Hertli2011}. Unfortunately
it is not possible to give an extensive overview for \ksat algorithms
here. We will however discuss the algorithm by Paturi, Pudl\'{a}k and
Zane in further detail.

For other circuit classes this question is still wide open. In
particular, for general polynomial size circuits no fast
satisfiability algorithms are known. If one believes that there is
such an algorithm, then one approach is to tackle more and more
general classes of circuits. If one believes that there are no fast
algorithms for general polynomial size circuits, then where is the
border between circuit classes that have fast satisfiability
algorithms and classes that don't? The \emph{Strongly Exponential Time
  Hypothesis} (\seth) \cite{ImpagliazzoPaturi1999} conjectures that,
while \ksat does have a fast satisfiability algorithm for any constant
$k$, many variants of \cnfsat do not. The hypothesis says that for
every constant $\mu > 0$, there is a $k$ such that \ksat cannot be
solved in time $O\left(2^{(1-\mu) n}\right)$. This would put \cnfsat
with no restriction on the width of the clauses or the size of the
formula on the side with no fast satisfiability algorithm.

For \cnfsat, there is a duality between the width of the clauses and
the size of the circuit \cite{CalabroImpagliazzoPaturi2006}. As a
consequence, a corollary of \seth is that there are no satisfiability
algorithms with constant savings for \cnf formulas with superlinear
size \cite{CalabroImpagliazzoPaturi2006}. It is therefore not
surprising that even beyond \cnf formulas, many satisfiability
algorithms with constant savings require the circuit to be linear
size.

In this survey we highlight several circuit classes that allow
satisfiability algorithms with constant savings. For a survey on fast
algorithms for $\np$-complete problems that goes beyond satisfiability
problems see the articles by Woeginger
\cite{Woeginger2003,Woeginger2008}.

Most restricted circuit classes were first defined in the context of
circuit lower bounds. In some sense, algorithms and lower bounds are
two sides of the same question. Algorithms try to find the most
efficient way to compute something, while lower bounds give
limitations on what is possible. This survey highlights several
results that show how advances on one problem can lead to advances in
the other. Intuitively, any satisfiability algorithm has to use the
structure of the circuit somehow to improve over exhaustive
search. The same structural properties can be used to argue which
functions a circuit class fails to represent.

Section \ref{sec:properties} discusses algorithms that rely on such a
property. Section \ref{sec:Williams} discusses a more direct
connection between satisfiability algorithms and lower bounds. In this
result by Williams \cite{Williams2010,Williams2011}, the lower bounds
are derived from satisfiability algorithms, where the algorithm is
treated as a black box. As a consequence, the connection between the
satisfiability algorithm and the lower bound does not directly rely on
a property of the circuit class, although the satisfiability algorithm
does. This result not only provides a blueprint for potentially
finding new lower bounds, but also describes the relationship between
algorithms and lower bounds in a more formal framework.

In Section \ref{sec:splitList} we discuss satisfiability algorithms
that use results on various polynomial time algorithms such as matrix
multiplication to get constant savings. Those algorithms exploit links
between satisfiability algorithms for circuits and problems within
computer science beyond circuits.

Santhanam \cite{Santhanam2012} wrote a survey with a similar focus as
this survey, discussing, among others, the results on \kcnf, DeMorgan
Formulas and $\acczero$ in greater detail.

\section{Preliminaries }
In this paper we consider different circuit classes. For a set of
functions $B$, a \emph{circuit} over basis $B$ is a sequence
$q_0,\ldots,q_{m}$ such that $q_0$ and $q_1$ are the constants $0$ and
$1$, $q_2$ to $q_{n+1}$ are the $n$ input variables $x_1,\ldots,
x_n$. For every index $j \geq n+2$ the circuit is defined by a gate,
which is defined by a function from $B$ and the indices of the inputs
to that function. The indices of the inputs are required to be less
than $j$. The semantics of a circuit is that at every position, the
function from $B$ is applied to the inputs, and given values to the
input variables, we can compute the value of every gate from left to
right. The last gate of a circuit is called the \emph{output
  gate}. The \emph{size} of a circuit is the number of gates other
than the constants and inputs.

A \emph{restricted circuit class} is a restriction of the definition
of the circuit is some way. Depending on the context, we refer to the
circuit class as either the set of circuits itself, or the class of
problems that can be decided with such a circuit. All restricted
circuit classes fix a basis $B$, and many restrict the size of the
circuit. For example, $\ppoly$ is the class of problems that can by
decided with a polynomial size circuit and basis $B_2$, which is the
set of all $8$ functions with fanin two. Some classes restrict the
\emph{depth} of a circuit, which is defined as the longest path from
the output gate to an input gate. Intuitively, the depth is the time
it takes to compute the result when running the computation with
maximal parallelization. Other circuit classes may restrict the
circuit to \emph{formulas}. In a formula, the output of a gate can be
used at most once as input to other gates. We do however allow that
literals and constants are used more than once. (Literals are
variables or their negation.)

One important class of circuits are \emph{formulas in conjunctive
  normal form} (\cnf). A \cnf is a circuit where the output gate is an
AND of clauses, and every clause is an OR of literals. \kcnf further
restricts the fanin of a clause to $k$.

$\aczero$ circuits allow AND and OR gates with arbitrary fanin, but only
constant depth. We can assume without loss of generality that an
$\aczero$ circuit consists of alternating layers of AND and OR gates, as
we can otherwise merge two gates into one.

$\acczero$ extends the basis of $\aczero$ circuits and also allows
$\text{MOD}_m$ gates for any $m>1$ for arbitrary fanin.

Another extension of $\aczero$ discussed in this paper is $\tczero$,
the class of constant depth circuits where the gates are threshold
gates, i.e. gates that take a weighted sum of the inputs and compare
it to a threshold. For a threshold gate, both the threshold and the
weights can be arbitrary real numbers. In this survey, we consider
only threshold circuits of depth two. We call the output gate the
\emph{top-level gate} and the other threshold gates the
\emph{bottom-level gates}. We call the number of variable occurrences
the \emph{number of wires}.

For the \maxksat problem we are given a \kcnf and a threshold $t$, and
need to decide if there is an assignment that satisfies at least $t$
clauses. We view it as the satisfiability problem on the circuit class
that consists of $k$-clauses and a threshold gate as the output gate.

A DeMorgan formula is a formula on AND and OR gates of fanin two. The
depth is unbounded for DeMorgan Formulas. Note that we do not need to
allow NOT gates other than when negating literals as we can push any
NOT gates to the literals using DeMorgan's Laws. For some arguments it
is more intuitive to view a DeMorgan formula as a binary tree where
the literals are leafs and the gates are inner nodes. With this in
mind, we refer to the gates that lead up to a gate $q_j$ as the
subtree rooted at $q_j$. We measure the size of a DeMorgan formula in
its \emph{leaf size}, the number of input literals.

Given a circuit $C$ on a variable set $V$, an \emph{assignment} is a
function $V \to \{0,1\}$. A \emph{restriction} is an assignment to a
subset of $V$. For a restriction $\rho$, we denote by $C|_{\rho}$ the
circuit $C$ where all variables restricted by $\rho$ are replaced by
constants accordingly.

For exponential functions $f(n)$, we use $\tilde{O}(f(n))$ as a
shorthand for $O(f(n))\text{poly}(n)$. For algorithms with runtime
$\tilde{O}\left(2^{(1-\mu) n}\right)$, we refer to $\mu$ as the
\emph{savings} of the algorithm. Typically, we are interested in
algorithms with constant savings, although the results in Section
\ref{sec:Williams} rely on algorithms with subconstant savings.

\section{Properties of Circuit Classes}
\label{sec:properties}
In this section we consider three satisfiability algorithms for
different circuit classes that are closely linked to lower bounds for
the same class. For \kcnf, we discuss a satisfiability algorithm and a
lower bound by Pudl\'{a}k, Paturi and Zane
\cite{PaturiPudlakZane1997}. For $\aczero$ we discuss an algorithm and
a lower bound by Impagliazzo, Matthews and Paturi
\cite{ImpagliazzoMatthewsPaturi2012}. Lastly, for DeMorgan Formulas we
discuss a lower bound by Subbotovskaya \cite{Subbotovskaya1961} and a
satisfiability algorithm by Santhanam \cite{Santhanam2010} that rely
on the same properties.

There is no clear relation if the algorithm follows from the lower
bound or vice versa. For some of these results, the results were
motivated by the need to find faster satisfiability algorithms, with
the lower bounds a consequence of the proof technique. In other cases,
the idea for the algorithm stems from the lower bound. In the case of
DeMorgan Formulas, the circuit lower bound precedes the algorithm by
almost 50 years.

The three algorithms are similar in several aspects. In all three
examples we can extract a property of the circuit class that neither
talks about satisfiability algorithms nor lower bounds, but both are
derived from the property without using any additional information
about the circuit class. For \kcnf, the statement is about how to
encode satisfying assignments, for $\aczero$ it is about partitions of
the hypercube and for DeMorgan Formulas it is about shrinkage under
restriction.

Furthermore, all three algorithms follow the same general outline:
Restrict a set of variables to constants, simplify the circuit, and
repeat until the output of the circuit is a constant. Where the
algorithms rely on properties of the underlying circuit is when
simplifying the restricted circuit. For \kcnf, the simplification step
removes variables beyond the ones already restricted. For $\aczero$,
the simplify step reduces the depth of the circuit, and for DeMorgan
Formulas, the simplify step shrinks the size of the circuit by a
nontrivial number of gates.

The basis of these algorithms is the \emph{random restriction
  technique}, which was first used by Subbotovskaya in her lower bound
for DeMorgan Formulas. In its basic form, one picks a random variable
(or a random set of variables) and restricts the variable randomly to
$0$ or $1$. However, not all algorithms that rely on the random
restriction technique are randomized algorithms.  Some algorithms like
the algorithm for \ksat are more natural as a randomized algorithms,
and additional ideas are required to derandomize it. Other algorithms
such as Santhanam's algorithm for DeMorgan Formulas explore all
possible random choices for a restriction. As a result, the algorithm
is deterministic.

\subsection{Satisfiability Coding Lemma}
\label{sec:PPZ}
The Satisfiability Coding Lemma \cite{PaturiPudlakZane1997} provides a
description of the space of satisfying assignments of a \kcnf in terms
of how many bits are required to describe an assignment.

Let $F$ be a formula in conjunctive normal form on $n$ variables. A
satisfying assignment $\alpha$ is \emph{isolated in direction} $x$ if
the assignment that only differs from $\alpha$ in variable $x$ does
not satisfy $F$. We say an assignment is isolated, if it is isolated
in all directions and $j$-isolated, if it is isolated in exactly $j$
directions. If $\alpha$ is isolated in direction $x$, then there must
be a clause $C$, such that $\alpha$ sets all literals of this clause
to false except the literal corresponding to $x$. We call such a
clause a \emph{critical clause}.

An encoding of a satisfying solution is an injective map from
satisfying solutions to binary strings. The simplest such encoding is
to fix a permutation of the variables, and then describe the
assignment to every variable as an $n$-bit string. However, since $F$
is a \cnf, we can do better. If there is a unit clause $x$ (i.e. a
clause consisting only of the literal $x$), then every assignment
satisfying $F$ must assign the value $1$ to $x$. Likewise, if there is
a clause $\overline{x}$, every satisfying assignment sets $x$ to
$0$. In either case we can simply omit $x$ in our description of the
satisfying assignment, and describe the assignment as an $(n-1)$-bit
string that still uniquely defines the satisfying
assignment. Furthermore, once we assign a value to $x$, the formula
simplifies. If we assign $1$ to $x$, clauses containing the literal
$x$ are satisfied and can be omitted. Clauses containing
$\overline{x}$ can be simplified by omitting the literal
$\overline{x}$, as this literal cannot be used to satisfy the clause
anymore. Therefore, by repeatedly either giving the assignment of the
next variable in the permutation or simplifying using a unit clause we
can encode any satisfying assignment. Algorithm \ref{alg:encode} gives
the full procedure to encode an assignment. It takes as input a
satisfying assignment $\alpha$ and a permutation $\pi$ on the
variables and produces an encoding $\text{enc}_{\alpha,\pi}$. The
encoding algorithm fails to produce a string if the input assignment
$\alpha$ does not satisfy $F$. Algorithm \ref{alg:decode} provides a
decoding algorithm for the opposite direction. It takes as input the
string $\text{enc}_{\alpha,\pi}$ and the permutation $\pi$ and outputs
the original assignment $\alpha$. It fails if the input is not a valid
encoding of a satisfying assignment.

\begin{algorithm}[t]
\label{alg:encode}
  \DontPrintSemicolon
  \KwData{formula $F$, number of variables $n$, permutation $\pi$,
    satisfying assignment $\alpha$}
  \KwResult{$\text{enc}_{\alpha,\pi}$ encoding $\alpha$}
  $\text{enc} = \lambda$\;
  \For{$i=1$ to $|V|$}{
   $x \gets \pi(i)$\;
   \If{$\{x\} \in F$ and $\{\overline{x}\} \in F$}{
     return FALSE
     }
    \If{$\{x\} \in F$}{
      $v \gets 1$\;
    } \ElseIf{$\{\overline{x}\} \in F$}{
      $v \gets 0$\;
    }
    \ElseIf{ $\{x\} \not\in F$ and $\{\overline{x}\} \not\in F$}{
      $v \gets \alpha(x)$\; 
      $\text{enc} \gets \text{enc} \circ v$\; 
     }
     $F \gets \text{simplify}\left(F|_{x = v}\right)$\;
   }
   \Return enc\;
 \caption{Encoding satisfying assignments}
\end{algorithm}
\begin{algorithm}[t]
  \label{alg:decode}
  \DontPrintSemicolon
  \KwData{formula $F$, number of variables $n$, permutation $\pi$, encoding $\text{enc}_{\alpha,\pi}$}
  \KwResult{$\alpha$}
  $\alpha \gets \emptyset$\;
  $j \gets 0$\;
  \For{$i=1$ to $|V|$}{
   $x \gets \pi(i)$\;
   \If{$\{x\} \in F$ and $\{\overline{x}\} \in F$}{
     return FALSE
     }
    \If{$\{x\} \in F$}{
      $v \gets 1$\;
    } \ElseIf{$\{\overline{x}\} \in F$}{
      $v \gets 0$\;
    }
    \ElseIf{ $\{x\} \not\in F$ and $\{\overline{x}\} \not\in F$}{
      \If{$j>|enc|$}{\Return FALSE}
      $v \gets \text{enc}_j$\;
      $j \gets j+1$\;                                      
     }
     $\alpha \gets \alpha \cup \{x = v\}$\;
     $F \gets \text{simplify}\left(F|_{x = v}\right)$\;
   }
   \Return $\alpha$\;
 \caption{Decode a string to an assignment}
\end{algorithm}

The length of the encoding depends on the number of unit clauses we
encounter while encoding. The Satisfiability Coding Lemma gives a
bound on the expected length if we pick the permutation uniformly at
random.

\begin{lemma}[Satisfiability Coding Lemma]
  Let $F$ be a \kcnf and let $\alpha$ be a $j(\alpha)$-isolated
  assignment. For a uniformly chosen permutation $\pi$ we have
\begin{equation*}
  E[|\text{enc}_{\alpha,\pi}|] \leq n - j(\alpha)/k
\end{equation*}
\end{lemma}
\begin{proof}
  Let $x$ be a variable such that $\alpha$ is isolated in direction
  $x$ and let $C$ be its critical clause. We can omit the bit
  describing the assignment for $x$ if $C$ is a unit clause when $x$
  occurs in the permutation. This happens exactly when when $x$ is the
  last variable of $C$ in the permutation. Since we choose the
  permutation uniformly at random, we can omit $x$ with probability
  $1/k$. By linearity of expectation we omit $j(\alpha)/k$ bits in
  expectation.
\end{proof}

We can turn the Satisfiability Coding Lemma into a lower bound. We
call a finite set of strings $S \subseteq \{0,1\}^*$
\emph{prefix-free}, if there are no two strings in $S$ such that one
is a prefix of the other. Note that for a fixed permutation, the
strings encoding satisfying assignments are prefix-free.

\begin{theorem}
  A \kcnf has at most $2^{n-n/k}$ isolated satisfying assignments.
\end{theorem}
\begin{proof}
  For each isolated assignment, the average code length is at most
  $n-n/k$. The same bound holds for the average code length over all
  permutations and satisfying assignments. Therefore, there is some
  permutation $\pi$ such that the average code length over all
  isolated assignments is $n-n/k$. Let $S$ be the set of all such
  codes given permutation $\pi$ and let $S' = \{s*^{n-|s|}\mid s\in
  S\}$, i.e. extend all strings to length $n$ by adding $*$. We can
  interpret strings in $S'$ as restrictions in the obvious way, where
  $*$ represents the free variables. Since $S$ is prefix-free, the
  restrictions in $S'$ are not overlapping. Since a restriction that
  leaves $l$ variables free covers $2^l$ assignments, we have
  $\sum_{s\in S} 2^{n-|s|} \leq 2^n$ and hence $\sum_{s\in S}
  2^{-|s|}\leq 1$. Therefore
  \begin{equation*}
    n-\frac{n}k \geq \sum_{s\in S} \frac{|s|}{|S|} = \sum_{s \in S}
    \frac{-\log 2^{-|s|}}{|S|} \geq -\log \left(\frac{\sum_{s\in S}
        2^{-|s|}}{|S|}\right) \geq \log(|S|)
  \end{equation*}
  using Jensen's Inequality. Hence $|S| \leq 2^{n-n/k}$.
\end{proof}

To get a satisfiability algorithm from the Satisfiability Coding
Lemma, consider the following algorithm, which we call the PPZ
algorithm: Guess a permutation and an $n$-bit string uniformly at
random and try to decode the string using the decode algorithm. Note
that the algorithm might guess a string that is longer than what is
actually read while decoding.

\begin{lemma}
  Let $F$ be a \kcnf and $\alpha$ be a $j(\alpha)$-isolated
  solution. The PPZ algorithm returns $\alpha$ with probability at
  least $2^{-n+j(\alpha)/k}$.
\end{lemma}
\begin{proof}
  The main observation is that given a permutation $\pi$, PPZ returns
  $\alpha$ if and only if the algorithm guesses all bits according to
  $\text{enc}_{\alpha,\pi}$ (plus potentially some
  additional bits). The probability for this event is
  $2^{-|\text{enc}_{\alpha,\pi}|}$. Hence
  \begin{align*}
    P(\text{PPZ returns } \alpha) &= \sum_{\pi} P(\text{PPZ returns }
    \alpha \mid \pi) \frac1{n!} = \sum_{\pi}
    2^{-|\text{enc}_{\alpha,\pi}|} \frac1{n!} \\
    &\geq 2^{\frac1{n!}  \sum_{\pi} -|\text{enc}_{\alpha,\pi}|} \geq
    2^{-n+j(\alpha)/k}
  \end{align*}
  using Jensen's Inequality. 
\end{proof}

If we have at least one isolated solution, the probability is at least
$2^{-(1-1/k)n}$. We show that this success probability holds in
general. Intuitively, if there are no isolated solution, then there
must be many solutions.

\begin{theorem}
  The PPZ algorithm finds a satisfying assignment with probability at
  least $2^{-(1-1/k)n}$. 
\end{theorem}
\begin{proof}
  Let $\text{sat}(F)$ be the set of satisfying assignments for $F$ and
  for $\alpha \in \text{sat}(F)$, let $j(\alpha)$ denote the degree of
  isolation. We first prove $\sum_{\alpha \in \text{sat}(F)}
  2^{-n+j(\alpha))} \geq 1$.

  Fix some permutation and for all $\alpha$ (satisfying or not), let
  $s(\alpha) \in \{0,1\}^n$ be the string describing the assignment
  according to the permutation. Note that this is not the same as the
  encoding of a satisfying assignment. Further, for satisfying
  assignments $\alpha$, let $s'(\alpha) \in \{0,1,*\}^n$ be
  $s(\alpha)$, where a position is replaced by $*$, if $\alpha$ is
  isolated in that direction. We interpret the string $s'(\alpha)$ as
  a restriction in the natural way and claim that every assignment
  $\beta \in \{0,1\}^n$ is covered by some restriction.

  Let $\beta$ be an arbitrary assignment and let $\alpha \in
  \text{sat}(F)$ be the satisfying assignment with the smallest
  Hamming distance to $\beta$, i.e. the two assignments differ on the
  smallest number of variables. The restriction $s'(\alpha)$ must have
  a $*$ on every position where $s(\alpha)$ and $s(\beta)$ differ, as
  otherwise there would be a satisfying assignment with a smaller
  Hamming distance to $\beta$. Hence every assignment $\beta$ is
  covered by at least one restriction, and therefore $\sum_{\alpha \in
    \text{sat}(F)} 2^{j(\alpha)} \geq 2^n$ and $2^{-n+j(\alpha))} \geq
  1$.

  The success probability of the PPZ algorithm is therefore lower
  bounded by
  \begin{align*}
    P(\text{PPZ returns some } \alpha \in \text{sat}(F)) &=
    \sum_{\alpha \in \text{sat}(F)} P(\text{PPZ returns } \alpha) \geq
    \sum_{\alpha \in \text{sat}(F)} 2^{-n+j(\alpha)/k} \\ &= 2^{-(1-1/k)n}
    \sum_{\alpha \in \text{sat}(F)} 2^{-(n-j(\alpha))/k} \\ &\geq 2^{-(1-1/k)n}
    \sum_{\alpha \in \text{sat}(F)} 2^{-n+j(\alpha)} \geq 2^{-(1-1/k)n}
  \end{align*}    
\end{proof}

By repeating the PPZ algorithm we get an algorithm that runs in time
$\tilde{O}\left(2^{(1-1/k)n}\right)$ and has an arbitrarily small
one-sided error.

\subsection{$\aczero$ Circuits}
\label{sec:aczero}
Impagliazzo, Matthews and Paturi \cite{ImpagliazzoMatthewsPaturi2012}
give a characterization of $\aczero$ circuits based on
restrictions. For every $\aczero$ circuit on $n$ variables with size
$cn$ and depth $d$ there is a partition of the hypercube into
restrictions such that the function described by the circuit is
constant for each restriction. The size of the partition is
$\tilde{O}\left(2^{(1-\mu_{c,d}) n}\right)$ for $\mu_{c,d} =
\frac1{O(\log c + d \log d)^{d-1}}$ and can be constructed with only
polynomial overhead over its size.

The main idea behind the proof is a depth reduction technique based on
H\aa stad's \emph{Switching Lemma} \cite{Hastad1987}. The Switching
Lemma says if you take a \cnf formula (or, symmetrically, a \dnf) and
apply a random restriction to its inputs, then with high probability
you can represent the resulting function as a small \dnf formula (or
\cnf). This method can be applied for depth reduction. Given an
$\aczero$ circuit with alternating AND and OR gates, apply the
Switching Lemma to the lowest two levels of the circuit. After
applying a random restriction, with high probability we can swap the
bottom two layers. We then have two consecutive AND (or OR) layers,
which can be combined. The result is a circuit with reduced depth. The
main technical obstacle here is that H\aa stad's original Switching
Lemma is not sufficient for the required savings. Instead, they prove
the \emph{Extended Switching Lemma}, which deals with the case of
switching several \cnf formulas on the same variables together.

A satisfiability algorithm follows immediately. Construct the
partition as above and check each partition if it is the constant $0$
or $1$. The savings of the resulting algorithm is then $\mu_{c,d}$.

Another immediate consequence of such a partition is a bound on either
the depth or the size required for parity. The only partition of the
parity function into restriction where the function is constant has to
restrict all $n$ variables. Hence the size of such a partition is
$2^n$. Solving the inequality $\tilde{O}\left(2^{(1-\mu_{c,d})
    n}\right) \geq 2^n$ for either the size or the depth gives that
every polynomial size $\aczero$ circuit requires at least depth
$\frac{\log n}{\log \log n} - o\left(\frac{\log n}{\log \log
    n}\right)$ and any depth $d$ circuit requires at least
$2^{\Omega\left(n^{\frac1{d-1}}\right)}$ gates. These bounds match
lower bounds derived from the original Switching Lemma by H\aa stad
\cite{Hastad1987}, which is not surprising given that the techniques
are strongly related.

Another lower bound that follows from this partition is a bound on the
\emph{correlation} of parity with an $\aczero$ circuit. For a circuit
$C$ and a function $f$, the correlation is defined as 
\begin{equation*}
  P(C(\alpha) = f(\alpha)) - P(C(\alpha) \neq f(\alpha))
\end{equation*}
where we choose the assignment $\alpha$ uniformly at random.

Consider an arbitrary $\aczero$ circuit $C$ and its partition into
$\tilde{O}\left(2^{(1-\mu_{c,d}) n}\right)$ restrictions. Any
restriction that contains more than one assignment agrees with parity
on exactly half of all values. Hence its contribution to the
correlation is $0$. On the other hand, a restriction to a single
assignment contributes only $2^{-n}$ to the correlation. Hence the
correlation between $C$ and parity is at most
\begin{equation*}
  2^{-n} \tilde{O}\left(2^{(1-\mu_{c,d}) n}\right) = \tilde{O}\left(2^{-\mu_{c,d} n}\right)
\end{equation*}

\subsection{DeMorgan Formulas}
\label{sec:deMorgan}
In 1961, Subbotovskaya \cite{Subbotovskaya1961} gave a lower bound on
the size of DeMorgan Formulas computing parity based on a random
restriction technique. Her result is credited as the first use of a
random restriction technique, which is used in many results after her,
including the satisfiability algorithm for \kcnf in Section
\ref{sec:PPZ}, for $\aczero$ circuits in Section \ref{sec:aczero}, and
for depth two threshold circuits in Section \ref{sec:threshold}.

Consider an arbitrary DeMorgan formula on $n$ variables and size
$s$. Pick a random set of $n-k$ variables and set them uniformly at
random to either $0$ or $1$, how many gates are still required for the
resulting function? Subbotovskaya's argument given below proves that
the number of gates shrinks to at most $\left(\frac{k}{n}\right)^{3/2}
s$ in expectation, giving a \emph{shrinkage exponent} of at least
$1.5$. It immediately follows that parity requires $O(n^{1.5})$ gates,
as the number of gates required after fixing $n-1$ variables to
constants is $1$. Andreev \cite{Andreev1987} (see also
\cite{Jukna2012}) later constructed a function that requires
$O(n^{2.5})$ gates. It follows from this construction that the same
function requires at least $n^{\omega + 1}$ gates where $\omega$ is
the shrinkage exponent. Ending a line of research to improve the lower
bound on the shrinkage exponent
\cite{ImpagliazzoNisan1993,PatersonZwick1993}, H\aa stad
\cite{Hastad1998} proves that the shrinkage exponent is $2$.

The main observation is the following. Consider an AND gate $q_i$ such
that the literal $x$ (i.e. the non-negated variable) is a direct
input. Let the gate $q_j$ be the other input, called the neighbor. If
the random restriction sets $x$ to $0$, then the AND gate always
evaluates to false. Since we have a formula, the output of the circuit
does not depend on $q_j$ anymore. 

As a consequence, if $q_j$ depends on the variable $x$, then we can
simplify the circuit. Since the output only depends on the value of
$q_j$ if $x$ is $1$, we can replace every occurrence of $x$ in the
subtree rooted at $q_j$ with $1$. For the rest of this section we can
therefore assume w.l.o.g. that the neighbor of a literal $x$ or
$\overline{x}$ does not contain the variable $x$. In a lower bound
argument, we can argue that if the neighbor depends on $x$, then the
formula cannot be minimal. From an algorithmic standpoint, we can
argue that if there are occurrences of the variable in the neighbor,
then we can simplify the circuit in polynomial time.

If the gate is an OR gate or the direct literal is negated then the
case is symmetric. Note that for this observation it crucial that
DeMorgan Formulas do not allow XOR gates, as we can otherwise fix a
direct input to either constant and the output still depends on the
neighbor.

Using the fact that we pick the restriction randomly we get the
following lemma.

\begin{lemma}
  Let $F$ be a minimal DeMorgan formula on $n$ variables with size
  $s$. If we restrict a random set of $n-k$ variable to $0$ or $1$
  uniformly at random, then the resulting formula has size at most
  $\left(\frac{k}{n}\right)^{3/2} s$ in expectation.
\end{lemma}
\begin{proof}
  First consider the special case where we restrict only one randomly
  chosen variable $x$. At the very least, all occurrences of $x$
  disappear. However, with a probability of $\frac12$ we can also
  remove the neighbor of $x$. Since the neighbor does not depend on
  $x$ by assumption, it must contain at least one occurrence of
  another variable. In expectation, $x$ feeds into $\frac{s}n$
  gates. Hence the expected number of leafs that we remove is $3/2
  \frac{s}n$ and the expected size of the remaining tree is at
  most $$s - \frac{3s}{2n} = \left(1 - \frac3{2n}\right) s \leq
  \left(\frac{n-1}n\right)^{3/2}s$$

  Restricting a random variable $n-k$ times therefore gives a formula
  with size at most 
  $$
  \left(\frac{n-1}n\right)^{3/2}\cdot
  \left(\frac{n-2}{n-1}\right)^{3/2}\cdots
  \left(\frac{k}{k+1}\right)^{3/2} s = \left(\frac{k}n\right)^{3/2} s 
  $$
\end{proof}

Santhanam uses the same ideas for a satisfiability algorithm. Let $F$
be a DeMorgan formula with size $cn$ for some constant $c$. While the
lower bound result considers random restrictions, the satisfiability
algorithm is deterministic. First of all, we simplify the formula as
before. We can remove gates with at least one constant input. For some
constants, we can remove the neighbor of the constant. Furthermore,
for every variable $x$ feeding directly into a gate, replace all
occurrences of $x$ in its neighbor with the appropriate
constant. After simplification, instead of restricting a random
variable, the algorithm restricts the variable that occurs the most
often. Since there are $cn$ leaves in total, we can always find a
variable $x$ that occurs at least $c$ times. Then recursively find a
satisfying assignment for $F|_{x=1}$ and $F|_{x=0}$ until the formula
has no inputs, i.e. is constant. The result is a partition into
restrictions where the circuit is constant, similar to the
satisfiability algorithm for $\aczero$ circuits.

To analyze the runtime of this algorithm, the shrinkage as given by
Subbotovskaya is not sufficient, as it only gives an expected
shrinkage. Instead, Santhanam gives a concentration bound for the
shrinkage. Consider the recursion tree of the algorithm, where every
vertex is labeled by some formula $F$. The children of the node
labeled $F$ are two nodes labeled $F|_{x=1}$ and $F|_{x=0}$,
simplified as described above. The leafs are labeled by constant
functions. The runtime of the algorithm is given, within a polynomial
factor, by the size of this tree. For every non-leaf node, for at
least one of its children, the size of the formula shrinks by $1.5 c$
gates, while for the other child, the formula shrinks by at least $c$
gates. We call the earlier child the ``good child'' and the other one
the ``bad child''. If we consider a randomly chosen path from the root
of the tree to a leaf, with high probability this path picks the
``good child'' several times. Any path that picks the ``good child''
often cannot be very long (i.e. close to $n$), as the tree must arrive
at a leaf when the number of gates shrinks to $1$. As a result, one
can derive a bound on the number of leaves of this tree.

The details of the calculation are omitted here. The savings of the
algorithm take the form $\frac1{\text{poly}(c)}$.

\section{From Satisfiability to Lower Bounds}
\label{sec:Williams}
In the examples of connections between satisfiability algorithms and
lower bounds discussed so far, the connection was implicit in
nature. There is no blueprint for extracting circuit lower bounds from
satisfiability algorithms or vice versa that follows directly from
these results. 

Williams \cite{Williams2010,Williams2011} gives a more formal
connection between satisfiability algorithms for a circuit class and
lower bounds for the same class. Given a satisfiability algorithm that
improves over brute force by only a superpolynomial amount, he
constructs a lower bound against $\nexp$ (nondeterministic exponential
time). Not only is the satisfiability algorithm used as a black box,
the result applies to a large set of natural circuit classes. By
giving a satisfiability algorithm for $\acczero$, Williams completes
an (unconditional) proof for $\nexp \not\subseteq \acczero$. Since the
connection between satisfiability algorithms and circuit bounds is
more general than just $\acczero$ circuits, his result opens up a
possible path to prove further lower bounds in the future.

The technique by Williams achieves a similar goal as the examples in
the previous section, as the result is both a satisfiability algorithm
for $\acczero$ and a lower bound for the same circuit class. The
satisfiability algorithms relies on properties of the circuit
class. However, instead of deriving a circuit lower bound directly
from the same properties, Williams adds another layer of
abstraction. The proof of the circuit lower bound does not depend on
the properties of the circuit directly, but only on the derived
satisfiability algorithm. As a consequence of this abstraction, he is
able to formalize a connection between algorithms and lower
bounds. While it is difficult to characterize what properties of
circuit classes lead to both satisfiability algorithms and lower
bounds, the abstraction allows a quantitative statement on the
required satisfiability algorithm.

In the first paper \cite{Williams2010}, Williams proves that if there
is an algorithm for general circuit satisfiability that improves over
exhaustive search by a superpolynomial amount, then $\nexp
\not\subseteq \ppoly$. The proof is an \emph{indirect diagonalization}
argument. Assuming $\nexp \subseteq \ppoly$ and the existence of a
fast satisfiability algorithm for general $\ppoly$ circuits, it gives
an algorithm to solve an arbitrary problem $L \in
\ntime\left(2^n\right)$ in nondeterministic time
$O\left(2^n/\omega\right)$ for some superpolynomial $\omega$. As a
result, there are no problems in $\ntime\left(2^n\right)$ that are not
in $\ntime\left(2^n/\omega\right)$, which contradicts the
nondeterministic time hierarchy theorem \cite{Cook1973,
  SeiferasFischerMeyer1978}.

For a rough outline of the proof, suppose there is a satisfiability
algorithm for general circuits that improves over exhaustive search by
a superpolynomial factor and $\nexp \subseteq \ppoly$. Then pick an
arbitrary problem $L$ in $\ntime\left(2^n\right)$ and reduce it to the
\emph{Succinct-\ksat[3]} problem, which is $\nexp$-complete. The
Succinct-\ksat[3] problem is a variation on \ksat[3] for exponential
formulas. Instead of having the \kcnf[3] as an direct input, the input
is a polynomial size circuit, such that on input $i$ in binary, the
output is the $i$th bit of the encoding of the \kcnf[3]. The
Succinct-\ksat[3] problem is then to decide if the implied \kcnf[3] is
satisfiable. By the $\nexp$-completeness of Succinct-\ksat[3] we can,
given an input $x$ to $L$ of length $n$, construct a polynomial size
circuit $C$ with $n + O(\log n)$ inputs such that on input $i$ in
binary, the output is the $i$th bit of a \kcnf[3] that is satisfiable
if and only if $x \in L$. The number of variables of this \kcnf[3]
formula is exponential in $n$.

To test the satisfiability of this circuit without explicitly writing
out the \kcnf[3] formula, we use the idea of a \emph{universal
  witness}. Impagliazzo, Kabarnets and Wigderson
\cite{ImpagliazzoKabanetsWigderson2002} show that if $\nexp \subseteq
\ppoly$, then for every satisfiable instance of a Succinct-\ksat[3]
problem there is a polynomial size circuit such that on input $i$ in
binary, it outputs the value of the $i$th variable in a satisfying
assignment.

The nondeterministic algorithm proceeds as follows. First
nondeterministically guess the universal witness for the given
Succinct-\ksat[3] problem. Since the goal is to give an algorithm that
runs in $\ntime\left(2^n/\omega\right)$ the algorithm is free to use
nondeterminism at this point. Let this circuit be called $D$. From the
Succint-\ksat[3] instance $C$ we can construct a circuit $C'$ that
takes as input a number $i$ in binary, and outputs the $i$th clause,
consisting of three variables in binary (requiring $n + O(\log n)$
bits each) and three bits to indicate if the literals are
negated. Each of these variables is then given as input to the circuit
$D$. As a last step, we can check if the values that $D$ assigns to
the variables satisfies the clause.

The circuit $D$ is a universal witness for the \kcnf[3] formula if and
only if the constructed circuit is unsatisfiable, i.e. there is no
input $i$ such that the universal witness does not give an assignment
that satisfies the $i$th clause. Using the assumed fast algorithm for
circuit satisfiability, we can decide this in time
$O\left(2^n/\omega\right)$, resulting in an overall algorithm in
$\ntime\left(2^n/\omega\right)$, contradicting the nondeterministic
time hierarchy theorem.

In the second paper \cite{Williams2011}, Williams refines his result
for restricted circuit classes. For any circuit class $\mathcal{C}$
that contains $\aczero$ and is closed under composition, if there is a
satisfiability algorithm for $\mathcal{C}$ that improves over
exhaustive search by a superpolynomial amount, then $\nexp
\not\subseteq \mathcal{C}$. The main part of the proof is ensuring
that the circuit constructed for the proof is in the class
$\mathcal{C}$ so that we can apply the supposed algorithm for
$\mathcal{C}$-SAT. In particular, the circuit $C'$ that takes as input
a value $i$ and returns the $i$th clause is not necessarily in the
class $\mathcal{C}$. The key idea to get around this is by guessing
and checking an equivalent $\mathcal{C}$-circuit, and then building
the whole circuit using the guessed component. By giving an algorithm
for $\acczero$-SAT in the same paper he completes the proof for $\nexp
\not\subseteq \acczero$.

If this approach is useful for other circuit classes than $\acczero$
depends on if it is possible to find fast satisfiability algorithms
for these classes. The result does certainly motivate the search for
satisfiability algorithms for circuit classes that sit in expressive
power somewhere between $\acczero$ and $\ppoly$. 

\section{Reductions to Polynomial Time Problems}
\label{sec:splitList}
Lower bounds and algorithms faster than the trivial approach are not
something unique to circuits. For example, matrix multiplication has a
trivial $O(n^3)$ algorithm. Until Strassen \cite{Strassen1969} gave a
faster algorithm, it was unknown if this is the best we can do. Since
then, several algorithms were discovered that improve on Strassen's
runtime, most notably Coppersmith and Winograd
\cite{CoppersmithWinograd1987}, and Williams
\cite{VVWilliams2012}. Despite this progress, the exact value for the
\emph{matrix multiplication exponent}, the smallest $\omega$ such that
matrix multiplication can be solved in time $O(n^{\omega})$ is still
unknown. It follows from Williams' result that $\omega < 2.3727$, but
it is not clear how far this can be improved.

The ingenuity that goes into these faster algorithm can be used for
faster satisfiability algorithms. Williams \cite{Williams2004} uses
this idea directly and reduces \maxksat[2] to matrix
multiplication. Would one use the reduction to matrix multiplication
and then use the trivial algorithm for the multiplication, the
resulting algorithm would run in time $\tilde{O}\left(2^n\right)$. It
is the faster matrix multiplication algorithm that results in constant
savings.

The other examples discussed here reduce a satisfiability problem to
other polynomial time problems. Impagliazzo, Paturi and Schneider
\cite{ImpagliazzoPaturiSchneider2013} give a satisfiability algorithm
for depth two threshold circuits that reduces the problem to the
\emph{Vector Domination Problem}, the problem of finding two vectors
such that one dominates the other on every coordinate. The problem can
be trivially solved in quadratic time. However, only using an
algorithm faster than quadratic for the vector problem do we get a
satisfiability algorithm with any savings.

Lastly we discuss a result by P\u{a}tras\c{c}u and Williams
\cite{PatrascuWilliams2010}, who reduce \cnfsat to
\emph{$k$-Dominating Set}, the problem of finding a set of at most $k$
vertices in a graph such that every vertex is either in the set or
adjacent to a vertex in the set. This reduction has a different flavor
from the other reductions in the conclusions we can draw from the
result. For both \maxksat[2] and threshold circuits, the result is an
algorithm with constant savings by using a fast algorithm for the
polynomial time problem. For \cnfsat, no algorithm with constant
savings is known. If one subscribes to the belief that there are no
algorithms with constant savings for \cnfsat, then the reduction gives
a lower bound for $k$-Dominating Set. If one does not believe that
such a lower bound exists, then the reduction gives a mean to find a
fast satisfiability algorithm.

All three algorithm follow a paradigm called ``Split and List'': Split
the variable set into several parts, and list every possible
restriction of one of the parts. Using the list of (exponentially
many) restrictions, the problem is then reduced to an exponentially
large instance of the underlying polynomial time problem. As a
consequence of the ``Split and List'' approach, all three algorithms
require exponential space, which is a limiting factor for using these
algorithms in practice.

The ``Split and List'' approach opens up a wide range of possible
algorithms to explore. While matrix multiplication is a well studied
problem with countless applications, the same is not true for the
Vector Domination Problem. I am not aware of any applications outside
of the literature on satisfiability algorithms, although it is not
unlikely that it was used (under a different name) in a different
context. This motivates looking for more problems that have not
gathered a lot of attention but might have both a ``Split and List''
reduction from satisfiability problems and a nontrivial algorithm. A
good place to start might be \emph{quadratic problems}, problems whose
trivial algorithm runs in quadratic time. The Vector Domination
Problem is an example. There are many problems where the goal is to
find a pair that satisfies some property and that have trivial
quadratic runtime. Just as there is no known characterization of which
circuit classes allow fast satisfiability algorithms, there is no
characterization of which quadratic problems have subquadratic
algorithms.

For more applications of the ``Split and List'' approach, see Chapter
6 of Williams' Ph.D. thesis \cite{Williams2007}. P\u{a}tras\c{c}u and
Williams \cite{PatrascuWilliams2010} also discuss further examples.

\subsection{\maxksat[2]}
\label{sec:max2sat}
In this section we consider an algorithm for \maxksat[2] by Williams
\cite{Williams2004}. Let $F$ be a \kcnf[2] on $n$ variables and $m$
clauses. The \maxksat[2] problem asks if given a threshold $t$, is it
possible to satisfy at least $t$ clauses. Williams gives an algorithm
with constant savings that also generalizes to a weighted version of
\maxksat[2], if the weights are small and integer. For the purpose of
this paper, we will consider the unweighted case only.

This algorithm does not generalize directly to \maxksat. There are no
known algorithms that achieve constant savings for \maxksat for $k\geq
3$. This marks a significant difference between \maxksat and \ksat, as
for \ksat there are algorithms achieving constant savings for all
constants $k$.

Using a ``Split and List'' technique, the algorithm reduces
\maxksat[2] to the problem of finding a triangle in a $2^{n/3}\times
2^{n/3} \times 2^{n/3}$ tripartite graph, which in turn can be solved
by multiplying two $2^{n/3}\times 2^{n/3}$ matrices.
Let $\omega$ denote the matrix multiplication exponent, i.e. the
exponent of the fastest possible matrix multiplication algorithm.

\begin{theorem}
  Let $F$ be a \kcnf[2] and let $t \in \nats$. There is an algorithm
  that to find an assignment that satisfies at least $t$ clauses and
  runs in time $\tilde{O}\left(2^{\frac{\omega}3 n}\right)$, where
  $\omega$ is the matrix multiplication exponent.
\end{theorem}
\begin{proof}
  We assume $n$ is divisible by $3$. Separate the set of variables
  into three sets $A$,$B$, and $C$ all of size $\frac{n}3$. We can
  then distinguish six types of $2$-clauses:
  \begin{enumerate}
  \item Both variables of the clause are in $A$.
   \item Both variables are in $B$.
   \item Both variables are in $C$.
   \item Exactly one variable is in $A$ and exactly one variable is in
     $B$.
   \item One variable is in $B$ and one variable is in $C$.
   \item One variable is in $A$ and one variable is in $C$.
  \end{enumerate}

  We say a clause is of type $T_a$, $T_b$, $T_c$, $T_{ab}$, $T_{bc}$
  or $T_{ac}$ respectively. 

  For numbers, $s_a$, $s_b$, $s_c$, $s_{ab}$, $s_{bc}$ and $s_{ac}$
  such that their sum is at least $t$, the algorithm decides if there
  is an assignment that satisfies exactly $s_D$ clauses of type $T_D$
  for $D \in \{a, b, c, ab, bc, ac\}$. Since each number $s_D$ is a
  number between $0$ and $m$, there are at most $m^6$ combinations of
  numbers. A \kcnf[2] has at most $4n^2$ clauses, hence solving each
  combination separately is only a polynomial overhead.

  We now construct the following graph. Let $V_A$ be the set of
  assignments to the variables $A$ such that it satisfies exactly
  $s_a$ variables of type $T_a$. Likewise, let $V_B$ and $V_C$ be the
  set of assignments that satisfy exactly $s_b$ and $s_c$ clauses of
  their respective type. The vertex set of the graph is $V = V_A \cup
  V_B \cup V_C$. We have an edge between a vertex in $V_A$ and a
  vertex in $V_B$ if the two assignments together satisfy exactly
  $s_{ab}$ clauses of type $T_{a,b}$. We add edges between $V_B$ and
  $V_C$, and $V_A$ and $V_C$ in a similar fashion.

  There is an assignment to the variables that satisfies exactly $s_D$
  clauses of type $T_D$ for all $D$, if and only if there is a
  triangle in the constructed graph. The assignment corresponds to the
  three vertices in the triangle.

  To find a triangle using matrix multiplication we construct a matrix
  $M_{ab}$ such that $M_{ab}[i,j] = 1$ if there is and edge between
  the $i$th element of $V_A$ and the $j$th element of $V_B$. We also
  construct matrices $M_{bc}$ and $M_{ac}$ in a similar fashion. Then
  there is a triangle if and only if there is an $i$ and a $j$ such
  that $\left(M_{ab}\cdot M_{bc}\right)[i,j] \geq 1$ and $M_{ac}[i,j]
  = 1$. Since all matrices have size at most $2^{n/3}\times 2^{n/3}$
  we can do the multiplication in time $O\left(2^{\frac{\omega}3
      n}\right)$. We have a multiplicative overhead as we have to do a
  matrix multiplication for every possible combination of numbers
  $s_D$. However, this overhead only contributes a polynomial factor
  to the time of the whole algorithm.
\end{proof}

\subsection{Threshold Circuits}
\label{sec:threshold}
In this section we consider threshold circuits of depth two. The
algorithm by Impagliazzo, Paturi, and Schneider
\cite{ImpagliazzoPaturiSchneider2013} combines random restrictions as
in Subbotovskaya's lower bound for DeMorgan Formulas with the ``Split
and List'' approach of Williams' \maxksat[2] algorithm.

We give an algorithm that decides satisfiability of a depth two
threshold circuit on $n$ variables with $cn$ wires and arbitrary real
weights that runs with savings of the form $\frac1{c^{O(c^2)}}$. For
this algorithm, the restriction on the size of the circuit is not on
the number of gates, but on the number of literal occurrences.

The algorithm proceeds in two steps. First, we reduce the
satisfiability problem on a depth two threshold circuit with $cn$
wires to (not too many) satisfiability problems on depth two threshold
circuits on $n'$ variables and $\delta n'$ bottom-level gates, where
$\delta$ is a small constant we can choose freely. For the circuit
with few bottom-level gates, we need to allow \emph{direct wires},
i.e. variables that directly feed into the top-level gate.

As a second step, we reduce the satisfiability problem on the
remaining circuit to a problem we call the \emph{Vector Domination
  Problem}. The Vector Domination Problem takes as inputs two sets of
$d$-dimensional real vectors $A$ and $B$ with $|A| + |B| = N$ and the
goal is to find a vector $a \in A$ and a vector $b\in B$ such that for
every coordinate $i$, $a_i \leq b_i$. The reduction follows the
``Split and List'' paradigm.

The reduction from the satisfiability problem of depth two threshold
circuits on $n'$ variables and $\delta n'$ bottom-level gates to the
Vector Domination Problem gives an instance with $|A| = |B| = 2^{n/2}$
and dimension $d = \delta n$. We could then solve the Vector
Domination problem with the trivial $O(N^2)$ algorithm. Unfortunately,
this would not give an algorithm faster than exhaustive
search. Instead, we give an algorithm faster than quadratic for
$\delta < 0.136$ which yields a satisfiability algorithm with constant
savings.

For the reduction from a threshold circuit with a linear number of
wires to a threshold circuit with few bottom-level gates, the idea is
to select an (as large as possible) set $S$, such that restricting all
variables not in $S$ results in a circuit with at most $\delta |S|$
bottom-level gates. The key observation is that gates that depend on
at most one variable in $S$ simplify to a constant or a direct wire to
the top-level gate after restriction, independent of the values the
restriction assigns to the variables. We omit the details of the
calculations here. It is possible to find a set $S$ with $|S| \geq
\frac{\delta}{c^{O(c^2)}}$ such that the circuits have at most $\delta
|S|$ bottom-level gates. The proof relies on random restrictions.

The reduction from the satisfiability problem of a depth two threshold
circuit on few bottom-level gates to the Vector Domination Problem is
by a ``Split and List'' approach. First, for all remaining
bottom-level gates, fix the output to either $0$ or $1$. There are
$2^{\delta |S|}$ such combinations. For every threshold gate, we can
express the condition that its output is $1$ or $0$ respectively as a
linear inequality. Given output values for the bottom-level gate, we
can also express the top-level gate as a linear inequality in the
input variables. We now reduce the resulting system of linear
inequalities to the Vector Domination Problem as follows. Split the
set of remaining variables $S$ into two sets $S_1 =
\{x_1,\ldots,x_{|S|/2}\}$ and $S_2 = \{x_{|S|/2+1},\ldots, x_{|S|}\}$
of equal size. A linear inequality of the form $\sum_{i=1}^{|S|} a_i
x_i \geq t$ is true if and only if $\sum_{i=1}^{|S|/2} a_i x_i \geq t
- \sum_{i=|S|/2+1}^{|S|} a_i x_i$. Hence we can list all possible
assignments to the variables in $S_1$ and calculate
$\sum_{i=1}^{|S|/2} a_i x_i$ for each of the $\delta |S| + 1$
inequalities. Likewise, calculate $t - \sum_{i=|S|/2+1}^{|S|} a_i x_i$
for each assignment to $S_2$ and each inequality. The system of
inequalities is then satisfied by an assignment to both $S_1$ and
$S_2$ if the vector of these values for the assignment to $S_1$
dominates the vector of values for the assignment $S_2$. The resulting
Vector Domination Problem has $N= 2\cdot2^{|S|/2}$ vectors and
dimension $\delta |S| + 1 \approx 2\delta \log N$.

The last part of the algorithm for depth two threshold circuits is an
algorithm for the Vector Domination problem. Let $A$ and $B$ be the
two sets vectors of dimension $d$. Let $N= |A| + |B|$. The algorithm
is faster than the trivial $O(N^2)$ for $d \leq 0.272 \log N$ and
works as follows. Let $m$ be the median of the first coordinates of
both $A$ and $B$ and split the sets $A$ and $B$ into sets $A^+$,
$A^=$, $A^-$, $B^+$, $B^=$, and $B^-$ depending on if the first
coordinate is larger, equal, or smaller than the median. Then for a
vector $a \in A$ to dominate a vector $b \in B$ either $a \in A^+$ and
$b \in B^+$, or $a \in A^-$ and $b \in B^-$, or $a \in A^+ \cup A^=$
and $b \in B^- \cup B^=$. In the first two cases, the number of
vectors can be at most half as we split at the median. In the last
case, we know that the first coordinate of $a$ dominates the first
coordinate of $b$. We can therefore recurse on vectors of dimension
$d-1$. Furthermore, we require time $O(N)$ to calculate the median and
split the sets $A$ and $B$. Hence the runtime of this algorithm for
$N$ vectors of dimension $d$ is bounded by the recurrence relation.
\begin{equation*}
  T(N,d) = 2T(N/2,d) + T(N,d-1) + O(N)
\end{equation*}
which solves to
\begin{equation*}
  T(N,d) = \binom{d + \log N +2}{d+1}O(N)
\end{equation*}

This runtime is $O\left(N^{2 - f(\delta)}\right)$ where $f(\delta) >
0$ for $\delta < 0.272$. This results in an algorithm for the
satisfiability problem for depth two threshold circuits on $n'$
variables with $\delta n'$ bottom-level gates that runs in time
$\tilde{O}\left(2^{(1-g(\delta))n'}\right)$, where $g(\delta) > 0$ if
$\delta < 0.099$. The stronger requirement for $\delta$ comes from the
additional overhead of guessing the output value for all bottom-level
gates.

Choosing an arbitrary value $\delta$ smaller than $0.099$ results in a
satisfiability algorithm for depth two threshold circuits with $cn$
wires that runs in time
$\tilde{O}\left(2^{(1-1/c^{O(c^2)}))n}\right)$.
 
\subsection{Reductions as Lower Bounds}
\label{sec:dominatingSet}
In this section we discuss a reduction from a satisfiability problem
to a polynomial time problem that is considerably different from the
reductions discussed above in the conclusion it allows.  In the
previous sections we got satisfiability algorithms with constant
savings by first reducing to a polynomial time problem and then
solving that problem in a nontrivial way.  In this section we discuss
a result by P\u{a}tras\c{c}u and Williams \cite{PatrascuWilliams2010}
which reduces \cnfsat to \emph{$k$-Dominating Set}. For $k$-Dominating
Set, we are given a graph on $n$ vertices and $m$ edges and find a set
$S$ of $k$ vertices, such that every vertex is either in $S$ or
adjacent to a vertex in $S$.

There are two ways to interpret the result. If you subscribe to the
belief that there are no algorithms for \cnfsat with constant savings,
then there are no algorithms faster than the currently known ones for
$k$-Dominating Set. On the other hand, if you believe there to be
faster \cnfsat algorithms, then these reductions provide a possible
way of finding such an algorithm.

The trivial algorithm for the $k$-Dominating Set problem is to
enumerate all $\binom{n}k$ sets of size $k$ and test each of them in
linear time. The resulting algorithm then runs in time
$O\left(n^{k+1}\right)$. Eisenbrand and Grandoni
\cite{EisenbrandGrandoni2004} give a faster algorithm that uses fast
matrix multiplication.

\begin{lemma}
  For $k\geq 7$, there is an algorithm for $k$-Dominating Set that
  runs in time $n^{k+o(1)}$. 
\end{lemma}
The details of the algorithm are omitted here.

While this algorithm improves over the trivial algorithm by almost a
linear factor, it still requires that we consider every possible set
of size $k$.

Assuming that there is an algorithm for $k$-Dominating Set that avoids
listing every possible set of size $k$, we construct an algorithm for
\cnfsat that achieves constant savings.

\begin{theorem}
  Assume there exists $k\geq3$ such that $k$-Dominating Set has an
  algorithm that runs in time $O\left(n^{k-\varepsilon}\right)$ for
  some $\varepsilon > 0$. Then there is an algorithm for \cnfsat that
  runs in time $\tilde{O}\left(2^{(1-\varepsilon/k)n}\right)$. 
\end{theorem}
\begin{proof}
  Fix $k\geq 3$ to the smallest value such that there is a fast
  algorithm for $k$-Dominating Set. We assume that $k$ divides $n$.

  Given a \cnf on $n$ variables and $m$ clauses, we construct a graph
  in a ``Split and List'' fashion very similar to the construction of
  the graph for the \maxksat[2] algorithm. Split the vertex set into
  $k$ sets of size $\frac{n}k$ each and for each assignment to one of
  sets, add a vertex to the graph. We further add edges such that each
  of the groups of $2^{n/k}$ vertices is a clique. For each clause we
  add one extra vertex which we connect to all vertices that
  correspond to partial assignments that satisfy the clause, i.e. the
  partial assignment assigns the value $1$ to at least one literal in
  the clause. Lastly we add one extra vertex to every clique, not
  connected to any clause. We call this vertex the \emph{dummy node}.

  Consider a $k$-dominating set $S$ for this graph. Since every dummy
  node is covered, there must be at least one vertex chosen from every
  clique. Since there are $k$ cliques, each clique must have exactly
  one element in $S$. Furthermore, for every clause, there must be
  vertex in $S$ that is connected to the clause. Hence the partial
  assignment represented by that vertex satisfies the
  clause. Therefore, the union of the partial assignments in the set
  $S$ is an assignment that satisfies every clause.

  The number of vertices in the graph is $k2^{n/k} + k + m$. By
  assumption we can solve the $k$-Dominating Set problem, and
  therefore the \cnfsat problem, in time
  $$O\left(\left(k2^{n/k} + k + m\right)^{k-\varepsilon}\right) =
  O\left(2^{(1-\varepsilon/k) n}\right) \text{poly}(m)$$
\end{proof}

{\bf Acknowledgment:} I thank Ramamohan Paturi for helpful comments on
an earlier draft. 

\bibliographystyle{plain}
\bibliography{../../papers/refs.bib}

\end{document}